\documentclass[11pt]{article}
\usepackage[T1]{fontenc}
\usepackage[utf8]{inputenc}
\usepackage{lmodern,amsmath,amssymb,amsthm,booktabs,array}
\usepackage[margin=27mm]{geometry}
\usepackage{tikz}
\usepackage{needspace}
\usetikzlibrary{arrows.meta}
\usepackage[colorlinks=true,allcolors=blue]{hyperref}
\hypersetup{pdftitle={The compress-with-another threshold of Szykula's Figure 3 family},pdfauthor={Qichao Wang},pdfsubject={Exact pair distances in synchronizing automata},pdfkeywords={synchronizing automata, compress-with-another threshold, pair automaton}}
\newtheorem{theorem}{Theorem}
\newtheorem{lemma}[theorem]{Lemma}
\newtheorem{proposition}[theorem]{Proposition}
\newtheorem{corollary}[theorem]{Corollary}
\newcommand{\dist}{\operatorname{dist}}
\newcommand{\M}[1]{M_{#1}}
\newcommand{\Ost}[1]{O_{#1}}
\newcommand{\pair}[2]{\{#1,#2\}}
\title{The compress-with-another threshold\protect\\ of Szyku\l a's Figure 3 family}
\author{Qichao Wang\\Hebei University of Technology}
\date{3 October 2026}
\begin{document}
\maketitle
\begin{abstract}
We give a self-contained pair-automaton proof of the exact compress-with-another threshold of the corrected Figure~3 family in Szyku\l a's survey of open problems. For every $p\geq3$, the automaton has $n=3p$ states and $\mu(q_0)=4p=4n/3$. The word $(ba)^p(ab)^p$ attains this value. Two entrance potentials and an excluded region yield the lower bound, with the endpoint exception at $p=3$ treated explicitly. A separate reset construction proves $\operatorname{rt}(\mathcal A_p)\leq3p^2+4p-1$ for every parameter. Reproducible computations verify the transition and entrance identities; the all-parameter results follow from the explicit proofs.
\end{abstract}

\section{Introduction}
A complete deterministic semiautomaton $(Q,\Sigma,\delta)$ is synchronizing if some word maps $Q$ to a singleton. Such a word is a reset word, and the reset threshold $\operatorname{rt}(\mathcal A)$ is its minimum length. We use right actions, writing $qw=\delta(q,w)$, and count word length in original letters. For a state $q$, put
\[
 \mu(q)=\min\{|w|:qw=rw\text{ for some }r\in Q\setminus\{q\}\}.
\]
The minimum is finite in a synchronizing automaton. It concerns one collision involving a specified initial state, whereas a reset word must merge all initial states.

Szyku\l a~\cite{szykula} asks whether $\mu(q)=O(n)$ in every strongly connected synchronizing $n$-state automaton. His Figure~3, attributed there to D\.{z}yga's 2018 master's thesis~\cite{dzyga}, gives a binary family on $3k+6$ states proposed to attain $4n/3$. Conjecture~8 asserts synchronization of the family and this specified-state threshold. For general background we refer to Volkov~\cite{volkov}; the \v Cern\'y conjecture~\cite{cerny}, asking for reset words of length at most $(n-1)^2$, remains open. Our result concerns only the completed interpretation of Figure~3 defined in this paper.

We use $p=k+2\geq3$ and state the corrected transitions in full. The initial $b$-edge and the terminal matching edges complete the published description. The word printed in Conjecture~8 of the published version, $(ab)^p a(ab)^p$, has $4p+1$ letters; the corrected witness $(ba)^p(ab)^p$ has $4p$ letters and merges $q_0$ with $q_2$. Our reference~\cite{szykula} is the EPTCS publication and public arXiv v1; all claims below concern the completed transition definition stated here.

The proof strategy is as follows. Both letters have rank $n-1$, with unique kernel pairs $K_a=\pair{q_4}{q_{2p+1}}$ and $K_b=\pair{q_3}{q_{2p}}$. A first collision involving $q_0$ therefore consists of a path from an initial pair $\pair{q_0}{q_i}$ to one of these pairs, followed by its collapsing letter. We compute the two exact distances. An invariant excludes large opposite-parity main-branch pairs. An entrance potential on a main-branch layer pays for entry into a reflected mixed-branch layer; a second potential delays arrival at $K_b$. A further region containing $K_a$ can only be entered from $K_b$. This is a symbolic lower bound for every word and every initial partner.

Independent computational and proof-oriented work on the same family appeared in September 2026~\cite{machina,modal12,modal30}. Szyku\l a has also informed the author of other independently communicated proofs. We therefore make no priority claim; the purpose of the present note is to give a concise, self-contained pair-automaton proof together with an explicit synchronization construction.

\section{The corrected Figure~3 family}
Let $\mathcal A_p$ have $Q=\{q_0,\ldots,q_{3p-1}\}$ and alphabet $\{a,b\}$. The complete definition is conveniently given by matchings and three exceptional edges. Letter $a$ fixes $q_0$, swaps $q_{2j-1}$ and $q_{2j}$ for $1\leq j\leq p-1$, and has
\[
 q_{2p-1}a=q_{2p},\qquad q_{2p}a=q_{2p+1},\qquad q_{2p+1}a=q_3.
\]
On $q_{2p+2},\ldots,q_{3p-1}$ it swaps successive pairs starting at the even index $2p+2$; an unmatched last state is fixed. Letter $b$ swaps $q_{2j}$ and $q_{2j+1}$ for $0\leq j\leq p-1$, has $q_{2p}b=q_2$, and swaps successive pairs on $q_{2p+1},\ldots,q_{3p-1}$ starting at $2p+1$, again fixing an unmatched last state. These rules also specify $q_0b=q_1$.

For the proof write
\[
 \M i=q_i\ (0\leq i\leq2p),\qquad
 \Ost h=q_{4p-h}\ (p+1\leq h\leq2p-1).
\]
The outer coordinate runs in the reverse direction. Away from the exceptions, on either branch letter $a$ increases an odd coordinate by one and decreases an even coordinate by one; $b$ does the opposite. At the outer bottom $h=p+1$, the decreasing letter fixes the state. It is $a$ when $p+1$ is even and $b$ when $p+1$ is odd. The remaining exceptional rules are
\begin{equation}\label{eq:exceptions}
 \M0a=\M0,\quad \M{2p}a=\Ost{2p-1},\quad
 \M{2p}b=\M2,\quad \Ost{2p-1}a=\M3.
\end{equation}

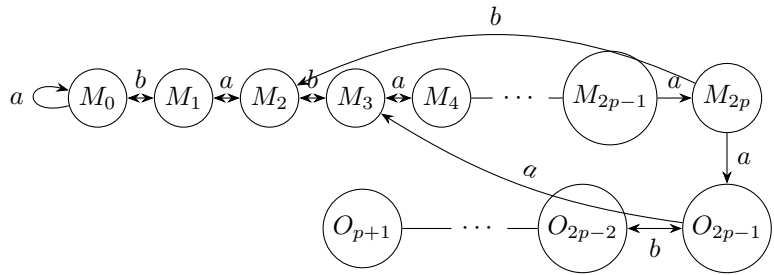
\begin{figure}[ht]
\centering
\begin{tikzpicture}[>=Stealth,scale=.91,every node/.style={font=\small},
 state/.style={draw,circle,inner sep=2pt}]
\node[state] (m0) at (0,0) {$M_0$};
\node[state] (m1) at (1.25,0) {$M_1$};
\node[state] (m2) at (2.5,0) {$M_2$};
\node[state] (m3) at (3.75,0) {$M_3$};
\node[state] (m4) at (5,0) {$M_4$};
\node (dots) at (6.15,0) {$\cdots$};
\node[state] (mpm) at (7.45,0) {$M_{2p-1}$};
\node[state] (mp) at (9.15,0) {$M_{2p}$};
\draw[<->] (m0)--node[above]{$b$}(m1);
\draw[<->] (m1)--node[above]{$a$}(m2);
\draw[<->] (m2)--node[above]{$b$}(m3);
\draw[<->] (m3)--node[above]{$a$}(m4);
\draw (m4)--(dots);\draw (dots)--(mpm);
\draw[->] (mpm)--node[above]{$a$}(mp);
\node[state] (ot) at (9.15,-1.9) {$O_{2p-1}$};
\node[state] (on) at (7.05,-1.9) {$O_{2p-2}$};
\node (od) at (5.55,-1.9) {$\cdots$};
\node[state] (ob) at (3.85,-1.9) {$O_{p+1}$};
\draw[->] (mp)--node[right]{$a$}(ot);
\draw[<->] (ot)--node[below]{$b$}(on);
\draw (on)--(od);\draw (od)--(ob);
\draw[->] (ot) to[bend left=12] node[above,sloped]{$a$}(m3);
\draw[->] (mp) to[bend right=25] node[above]{$b$}(m2);
\draw[->] (m0) to[loop left] node[left]{$a$}(m0);
\end{tikzpicture}
\caption{Schematic redrawn from the corrected transition description of Szyku\l a's Figure~3 family~\cite{szykula}, for $p\geq4$. Unlabelled chain segments continue the alternating matchings. The outer bottom is fixed by its decreasing letter; for $p=3$ the outer branch consists only of $O_5,O_4$. The complete rules are given in the text.}
\end{figure}

Both letters have rank $3p-1$. Indeed, their only double fibers are
\begin{equation}\label{eq:kernels}
 a^{-1}(\M3)=\pair{\M4}{\Ost{2p-1}}=K_a,
 \qquad b^{-1}(\M2)=\pair{\M3}{\M{2p}}=K_b.
\end{equation}
Each letter misses one image state: $a$ misses $\M{2p-1}$ and $b$ misses $\M{2p}$. All other fibers are singletons.
The underlying directed graph is strongly connected. The main matchings connect $\M0,\ldots,\M{2p-1}$ in both directions, the final main $a$-edge reaches $\M{2p}$, and $a$ then enters the outer branch. Its matchings connect every outer state in both directions, and its top returns to $\M3$; $\M{2p}b=\M2$ also returns to the main chain.

\section{Exact pair distances and the threshold}
The unordered pair graph has vertex set $\binom Q2$. It has a directed edge $P\xrightarrow{x}Px$ for $x\in\{a,b\}$ whenever $|Px|=2$; collapsing edges are omitted. Set
\[
 S_p=\{\pair{\M0}{q}:q\in Q\setminus\{\M0\}\}.
\]
Distance from $S_p$ means the minimum over its vertices. A first collision must use a kernel pair, so
\begin{equation}\label{eq:reduction}
 \mu(q_0)=1+\min\{\dist(S_p,K_a),\dist(S_p,K_b)\}.
\end{equation}

\begin{lemma}\label{lem:dist}
For every $p\geq3$,
\[
 \dist(S_p,K_b)=4p-1,\qquad \dist(S_p,K_a)=4p.
\]
\end{lemma}
We prove the lower bounds first. All paths considered below start in $S_p$ and remain in the distinct-pair graph. No bound on path length or restriction on the choice of letters is imposed.

\paragraph{The index bound and excluded region.}
Every original state index increases by at most one in one letter. The particle that started at $q_0$ consequently has index at most the elapsed time $t$. For a current pair $P$ this gives
\begin{equation}\label{eq:index}
 t\geq\min\{i:q_i\in P\}.
\end{equation}
Let $J$ consist of opposite-parity main pairs whose coordinate sum exceeds $2p+1$. The only reachable member of $J$ is $K_b$.
To see this, an ordinary matching step on an opposite-parity pair preserves its sum, and an ordinary step on a same-parity pair preserves that parity type. A pair involving the fixed $\M0$ cannot acquire a sum above $2p+1$. The $b$-jump from $\M{2p}$ lands at $\M2$; its other main image is at most $\M{2p-1}$. Hence the resulting main-pair coordinate sum is at most $2+(2p-1)=2p+1$, so it cannot enter $J$, which requires a strictly larger sum. The $a$-step from $\M{2p}$ leaves the main branch.

The terminal edge $\M{2p-1}a=\M{2p}$ also requires attention. A target in $J$ containing $\M{2p}$ has the form $\pair{\M{2p}}{\M j}$ with odd $j\geq3$. For $5\leq j\leq2p-3$, its only $a$-preimage is $\pair{\M{2p-1}}{\M{j+1}}$, already in $J\setminus\{K_b\}$ since its sum is $2p+j>2p+1$. The remaining odd endpoint $j=2p-1$ has no $a$-preimage. For $j=3$, the preimage using $\M4$ is likewise already in $J\setminus\{K_b\}$, while $\pair{\M{2p-1}}{\Ost{2p-1}}a=K_b$. An outer-to-main step lands at $\M3$, and an opposite-parity sum above $2p+1$ requires its other image to be $\M{2p}$, exactly the case just treated. Two outer states cannot both enter the main branch in one letter. Finally, $K_ba=K_a$ and $K_bb$ collapses. No initial pair lies in $J$, and no edge from outside $J\setminus\{K_b\}$ enters that set. This proves the invariant by induction. Call a pair \emph{admissible} if it is outside $J\setminus\{K_b\}$.

\Needspace{8\baselineskip}
\paragraph{The first entrance potential.}
Put
\[
 C=\{\pair{\M x}{\M{2p+1-x}}:1\leq x\leq p\},
 \qquad \phi_C(x)=2p-2+|x-3|.
\]
No initial pair lies in $C$. Each internal edge changes the smaller coordinate by at most one, so the potential increases by at most one. The complete list of admissible entrances from outside $C$ is
\begin{center}\small
\begin{tabular}{@{}cllc@{}}\toprule
Letter & Source & Target & Source-time bound\\\midrule
$a$ & $\pair{\M{2p-3}}{\Ost{2p-1}}$ & $\pair{\M3}{\M{2p-2}}$ & $2p-3$\\
$b$ & $\pair{\M{2p-2}}{\M{2p}}$ & $\pair{\M2}{\M{2p-1}}$ & $2p-2$\\\bottomrule
\end{tabular}
\end{center}
Here and below source-time bounds concern arrival before the entering letter. Both follow from~\eqref{eq:index}; after adding one they equal the target potential. For completeness of the list, an ordinary inverse matching preserves the opposite-parity sum, hence gives a source in $C$. The exceptional $a$-fiber at $\M3$ allows the first outer source displayed above. The exceptional $b$-fiber at $\M2$ allows the second source. Targets containing $\M{2p}$ have no $b$-preimage, and those containing $\M{2p-1}$ have no $a$-preimage. The fixed $\M0$ cannot occur in $C$. These exhaust the exceptional inverse fibers. Thus induction on the path, including any exits and reentries, proves
\begin{equation}\label{eq:C}
 P\in C\ \Longrightarrow\ t\geq\phi_C(P),
 \qquad \dist(S_p,\pair{\M1}{\M{2p}})\geq2p.
\end{equation}

\paragraph{The reflected entrance potential.}
Consider the disjoint tracks
\[
 A=\{\pair{\M x}{\Ost{2p+1-x}}:2\leq x\leq p\},\qquad
 E=\{\pair{\M x}{\Ost x}:p+1\leq x\leq2p-1\}.
\]
For $H=A\cup E$ set $\phi_H=2p-1+x$, where $x$ is the main coordinate. No initial pair lies in $H$. Ordinary internal edges change $x$ by one. At the outer bottom, its decreasing letter exchanges $A$ at $x=p$ with $E$ at $x=p+1$; this also changes the potential by one. All admissible entrances from outside $H$ are
\begin{center}\small
\begin{tabular}{@{}cllc@{}}\toprule
Letter & Source & Target & Source-time bound\\\midrule
$a$ & $\pair{\M1}{\M{2p}}$ & $\pair{\M2}{\Ost{2p-1}}$ & $2p$\\
$b$ & $\pair{\M{2p}}{\Ost{2p-2}}$ & $\pair{\M2}{\Ost{2p-1}}$ & $2p$\\
$a$, $p\geq4$ & $\pair{\Ost{2p-1}}{\Ost{2p-3}}$ & $\pair{\M3}{\Ost{2p-2}}$ & $2p+1$\\
$a$, $p=3$ & $\pair{\Ost5}{\Ost4}$ & $\pair{\M3}{\Ost4}$ & $7$\\\bottomrule
\end{tabular}
\end{center}
The first bound is~\eqref{eq:C}; the others follow from~\eqref{eq:index}. Adding one gives $\phi_H$ at each target.

We check completeness uniformly. In the interior of $A$, the coordinates have opposite parity, so both forward and inverse matchings preserve their sum. In the interior of $E$, the equal coordinates move together. At the outer bottom the fixed coordinate joins precisely the two tracks just described. At the outer top, $a^{-1}(\Ost{2p-1})=\{\M{2p}\}$: target main coordinate $2$ gives the first entrance, whereas the top target in $E$ has main coordinate $2p-1$ and no $a$-preimage. Under $b$, the double fiber at $\M2$ gives the second entrance; its other source stays in $H$. To enter from two outer states under $a$, the main image must be $\M3$, hence the target is the $x=3$ member of $A$. Its ordinary outer preimage is $\Ost{2p-3}$ for $p\geq4$. For $p=3$, $h=4$ is the bottom and is fixed by $a$, so that preimage is instead $\Ost4$. These are the last two rows. The missing main images and the remaining ordinary inverse matchings give no further entrances. Therefore
\begin{equation}\label{eq:H}
 P\in H\ \Longrightarrow\ t\geq2p-1+x,
 \qquad \dist(S_p,\pair{\M{2p-1}}{\Ost{2p-1}})\geq4p-2.
\end{equation}
In particular the $p=3$ replacement pays the same potential, $8$, at its target; the top of $E$ is reached no earlier than time $10$.

\paragraph{Entrance into the kernel regions.}
There is only one admissible incoming edge to $K_b$:
\begin{equation}\label{eq:Kbentry}
 \pair{\M{2p-1}}{\Ost{2p-1}}\xrightarrow{a}K_b.
\end{equation}
Indeed, $b$ misses $\M{2p}$, and the other $a$-preimage is $\pair{\M4}{\M{2p-1}}$, in $J\setminus\{K_b\}$. Equation~\eqref{eq:H} now yields $\dist(S_p,K_b)\geq4p-1$.

To control $K_a$ define
\[
 B=\{\pair{\M x}{\Ost{2p+3-x}}:4\leq x\leq p+2\},\qquad
 F=\{\pair{\M{h+2}}{\Ost h}:p+1\leq h\leq2p-2\}.
\]
These tracks are disjoint. Their union $U$ contains $K_a$ and no initial pair. Its only admissible entrance from outside $U$ is $K_b\xrightarrow{a}K_a$. In fact, ordinary inverse matchings preserve the sum $2p+3$ in $B$, or the difference $2$ in $F$. At the outer bottom the decreasing letter joins $B$ at main coordinate $p+2$ to $F$ at main coordinate $p+3$. The only $U$-target containing the outer top is the $x=4$ member of $B$, namely $K_a$, whose $a$-preimage is $K_b$. A $b$-jump lands at $\M2$, and an outer-to-main $a$-jump at $\M3$, whereas every main coordinate in $U$ is at least $4$. Targets with main coordinate $2p$ have no $b$-preimage; those with main coordinate $2p-1$ have no $a$-preimage. The other top endpoint of $F$, with coordinates $(2p,2p-2)$, has its ordinary $a$-preimage in $F$ for $p\geq4$, and in $B$ for $p=3$. Thus the list exhausts both endpoints and all exceptional fibers, also when $p=3$ and $F$ has only one member.

Every path reaching $U$ has a first entrance, paid by a preceding visit to $K_b$. Hence its first entrance has length at least $4p$, and later visits cannot occur earlier. This proves $\dist(S_p,K_a)\geq4p$. A final $b$-collapse costs at least $4p$, and a final $a$-collapse at least $4p+1$.

\paragraph{Matching paths and the upper word.}
Put $c=ba$ and $u=ab$. Direct substitution in the transition rules gives cycles, in the displayed order,
\begin{equation}\label{eq:cycles}
 R=(\M0,\M2,\M4,\ldots,\M{2p},\M1)\quad\text{under }c,
 \qquad T=(\M0,\M1,\M3,\ldots,\M{2p-1},\M2)\quad\text{under }u.
\end{equation}
Both have length $p+2$. In particular $\pair{\M0}{\M2}c^p=\pair{\M{2p}}{\M1}$ and $\M1u^p=\M2$. The $u$-chain starting at $\M{2p}$ visits the even-indexed outer states upwards, reflects at the last original state, visits the odd-indexed outer states downwards, and then reaches $\M2$. It consists of $\M{2p}$ and all $p-1$ outer states, so its entrance length is $p$. Consequently
\begin{equation}\label{eq:upper}
 \pair{q_0}{q_2}(ba)^p(ab)^p=\{q_2\}.
\end{equation}
Immediately before the last $a$ in this word, the pair is $\pair{\M{2p-1}}{\Ost{2p-1}}$. After that $a$ it is $K_b$, and the final $b$ collapses it. These pairs are distinct until the last letter: an earlier collision would persist and contradict the distinct penultimate pair. Thus the prefix of length $4p-1$ reaches $K_b$ in the pair graph. Appending $a$, instead of the collapsing $b$, reaches $K_a$ at length $4p$. This proves both upper distance bounds and completes Lemma~\ref{lem:dist}.

\begin{theorem}\label{thm:main}
For every $p\geq3$, the corrected Figure~3 automaton $\mathcal A_p$ on $n=3p$ states satisfies
\[
 \mu(q_0)=4p=\frac{4n}{3}.
\]
The word $(ba)^p(ab)^p$ is a shortest word compressing $q_0$ with another state; it merges $q_0$ and $q_2$ at $q_2$.
\end{theorem}
\begin{proof}
Combine Lemma~\ref{lem:dist} with~\eqref{eq:reduction} and~\eqref{eq:upper}.
\end{proof}
Thus the numerical threshold assertion in the corrected form of Conjecture~8 holds for the completed Figure~3 family used in this paper.

\section{An explicit synchronization construction}
We give a reset construction separately, since a specified-state collision alone does not prove synchronization. Continue with $c=ba$, $u=ab$, and label the $c$-cycle $R$ in~\eqref{eq:cycles} by positions $0,\ldots,p+1$: position $0$ is $\M0$, position $i$ is $\M{2i}$ for $1\leq i\leq p$, and position $p+1$ is $\M1$. In these coordinates $c$ adds one modulo $r=p+2$.

The remaining main odd states descend by two under $c$ until reaching $\M1$. The outer $c$-chain visits odd original indices upwards, reflects at the terminal state, visits even original indices downwards, and then reaches $\M3$ and $\M1$. Its largest distance into $R$ is $p$. Since $c$ permutes $R$, this proves the exact image identity
\begin{equation}\label{eq:core}
 Qc^p=R.
\end{equation}
For $j\geq p$ the actual return word $d_j=ac^j$ acts on these positions by
\begin{equation}\label{eq:return}
 id_j=j-i\pmod r\quad(0\leq i\leq p-1),\qquad
 pd_j=(p+1)d_j=j+1\pmod r.
\end{equation}
For $1\leq i\leq p-1$, $\M{2i}a=\M{2i-1}$ reaches $\M1$ after $i-1$ applications of $c$, giving the first formula. Position $0$ is fixed by $a$. From position $p$, $a$ reaches the outer top, which enters $R$ at $\M1$ after $p$ applications of $c$. Position $p+1$ goes under $a$ to $\M2$, giving the same residue. This verifies every position in~\eqref{eq:return}.

Prepare the image with
\[
 F_p=c^p u^{p-1}a^4c^{p-2},\qquad |F_p|=6p-2.
\]
The image of $R$ under $u^{p-1}$ is
\[
 \{\M0,\M1,\M3,\ldots,\M{2p-1},\Ost{2p-1}\}.
\]
This follows from the $u$-cycle and the $u$-chain used above. Under $a^4$, the two states $\M{2p-1}$ and $\Ost{2p-1}$ both reach $\M4$, while all remaining displayed states are fixed. The resulting image is $\{\M0,\M1,\M3,\ldots,\M{2p-3},\M4\}$. Applying $c^{p-2}$ gives, in $R$-positions,
\begin{equation}\label{eq:prepare}
 QF_p=S_{p-2},\qquad S_t=\{0,\ldots,t\}\cup\{p,p+1\},\quad S_{-1}=\{p,p+1\}.
\end{equation}
Indeed the odd states give positions $p-4,p-5,\ldots,0,p+1$ when $p\geq4$, while $\M0,\M1,\M4$ give $p-2,p-3,p$. For $p=3$, the prepared set before the final $c$ is $\{\M0,\M1,\M3,\M4\}$; its image has positions $\{0,1,3,4\}=S_1$. Thus the same identity covers the smallest parameter.

For $0\leq t\leq p-2$, equation~\eqref{eq:return} gives
\begin{equation}\label{eq:shrink}
 S_td_{p+t}=S_{t-1}.
\end{equation}
For $2\leq t\leq p-2$, the interval $\{0,\ldots,t\}$ maps to $\{(p+t-i)\bmod r:0\leq i\leq t\}=\{p,p+1,0,\ldots,t-2\}$, where $r=p+2$. The pair $p,p+1$ merges to residue $(p+t+1)\bmod r=t-1$. Adding this image gives $\{p,p+1\}\cup\{0,\ldots,t-1\}=S_{t-1}$. At $t=1$, the interval maps to $\{p,p+1\}$ and the merged image is $0$, giving $S_0$. At $t=0$, the interval maps to $\{p\}$ and the merged image is $-1\bmod r=p+1$, giving $S_{-1}=\{p,p+1\}$. This checks the terminal steps as well as the general one.

\begin{proposition}\label{prop:sync}
For every $p\geq3$, an explicit reset word for $\mathcal A_p$ is
\begin{equation}\label{eq:reset}
 R_p=c^p u^{p-1}a^4c^{p-2}
       (ac^{2p-2})(ac^{2p-3})\cdots(ac^p)u^p.
\end{equation}
It maps $Q$ to $\{q_2\}$ and has length $3p^2+4p-1$. In particular, $\mathcal A_p$ is synchronizing and
\[
 \operatorname{rt}(\mathcal A_p)\leq3p^2+4p-1.
\]
\end{proposition}
\begin{proof}
Equation~\eqref{eq:prepare} initializes the image. Applying~\eqref{eq:shrink} in descending order $t=p-2,\ldots,0$ leaves positions $p,p+1$, that is $\pair{\M{2p}}{\M1}$. The final $u^p$ maps both to $\M2$. Counting every $a,b$ letter gives
\[
 |R_p|=(6p-2)+\sum_{j=p}^{2p-2}(2j+1)+2p
       =(6p-2)+(3p-1)(p-1)+2p=3p^2+4p-1.
\]
\end{proof}
Together with Theorem~\ref{thm:main}, this proves both intended assertions of the corrected form of Conjecture~8 for the completed Figure~3 family used in this paper.

\begin{corollary}
Suppose there are $c\in\mathbb R$, $n_0\in\mathbb N_{\geq1}$ and a single function $\varepsilon:\mathbb N_{\geq1}\to\mathbb R$ with $\varepsilon(n)/n\to0$ such that, for every $n\geq n_0$, every binary strongly connected synchronizing $n$-state automaton whose letters have corank one, and every state $q$, one has $\mu(q)\leq cn+\varepsilon(n)$. Then $c\geq4/3$.
\end{corollary}
\begin{proof}
The automata $\mathcal A_p$ belong to this class. For all sufficiently large $p$, the assumed uniform bound gives $4p\leq3pc+\varepsilon(3p)$, hence $4/3\leq c+\varepsilon(3p)/(3p)$. Letting $p\to\infty$ proves the claim.
\end{proof}

\section{Verification}
The independent checker reconstructs the letter matchings and verifies the excluded region, entrance lists, potential inequalities, upper witness, and reset-image identities for the $67$ parameters $p=3,\ldots,64$ and $100,128,256,512,1024$, up to $n=3072$. Forward-star and reverse-pair distances are also checked for $p=3,\ldots,12$, and all words of length below $12$ are checked at $p=3$. Supplementary checks compute both kernel distances and apply expanded reset words for $p=3,\ldots,24$. Reproduction commands, detailed operation counts and validation reports are supplied with the source. These computations verify the symbolic argument; the all-parameter proofs are those above.

The abstract layered-pair barrier mechanism underlying the entrance induction has been formalized in Lean~4.19.0. Its graph and entrance conditions are hypotheses: the concrete Figure~3 transitions, excluded region and all-parameter entrance classification have not been instantiated in Lean. The formalization therefore certifies the conditional mechanism, rather than the concrete theorem of this note.

\Needspace{8\baselineskip}
\paragraph{Acknowledgements and computational and AI assistance.}
The author thanks Marek Szyku\l a for correspondence. OpenAI language-model tools, including Codex, assisted exploratory proof search, computational checking, critique, and language editing. The mathematical claims are supported by the explicit proofs in this paper and reproducible checks; responsibility for the statements rests with the author.


\end{document}